\documentclass[a4paper,UKenglish]{lipics-v2018}

\usepackage{microtype}

\title{Adapting Local Sequential Algorithms to the Distributed Setting}

\titlerunning{}

\author{Ken-ichi Kawarabayashi}{National Institute of Informatics, Tokyo, Japan}{k\_keniti@nii.ac.jp}{}{}

\author{Gregory Schwartzman}{National Institute of Informatics, Tokyo, Japan}{greg@nii.ac.jp}{}{}

\authorrunning{K. Kawarabayashi and G. Schwartzman}

\Copyright{Ken-ichi Kawarabayashi and Gregory Schwartzman}

\subjclass{Theory of computation$\rightarrow$Distributed algorithms}

\keywords{Distributed, Approximation Algorithms, Derandomization, Max-Cut}

\category{}

\relatedversion{https://arxiv.org/abs/1711.10155}

\supplement{}

\funding{This work was supported by JST ERATO Grant Number JPMJER1201, Japan}

\acknowledgements{We would like to thank Ami Paz and Seri Khoury for many fruitful discussions and useful advice.}

\EventEditors{John Q. Open and Joan R. Access}
\EventNoEds{2}
\EventLongTitle{42nd Conference on Very Important Topics (CVIT 2016)}
\EventShortTitle{CVIT 2016}
\EventAcronym{CVIT}
\EventYear{2016}
\EventDate{December 24--27, 2016}
\EventLocation{Little Whinging, United Kingdom}
\EventLogo{}
\SeriesVolume{42}
\ArticleNo{23}
\nolinenumbers 
\hideLIPIcs  

\usepackage[usenames,dvipsnames,svgnames,table]{xcolor}
\usepackage[linesnumbered,vlined]{algorithm2e}
\usepackage{amsmath}
\usepackage{amsfonts, amssymb}
\usepackage{cite}
\usepackage{mathtools}

\usepackage{booktabs}
\usepackage{amsthm}

\newcommand{\etal}{\textit{et al.}\xspace}
\newenvironment{lemma-repeat}[1]{\begin{trivlist}
\item[\hspace{\labelsep}{\bf\noindent Lemma \ref{#1} }]\em }%
{\end{trivlist}}
\newenvironment{theorem-repeat}[1]{\begin{trivlist}
\item[\hspace{\labelsep}{\bf\noindent Theorem \ref{#1} }]\em }%
{\end{trivlist}}

\newcommand{\remove}[1]{}

\DeclareMathOperator*{\argmax}{argmax} 
\newcommand{\size}[1]{\ensuremath{\left|#1\right|}}
\newcommand{\set}[1]{\left\{ #1 \right\}}
\DeclarePairedDelimiter\ceil{\lceil}{\rceil}

\newcommand{\Xbar}{\overline{X}}
\newcommand{\Xnv}{L_v[\overline{X}]}

\newcommand{\LemCutsLocal}
{

The utility functions for Max $k$-Cut, Max-DiCut and max-agree correlation clustering with 2 clusters are local utility functions.
}

\newcommand{\ThmKuhnExt}
{
	For any constant $\epsilon \in (0,1/e)$ a weighted $\epsilon$-defective $O(\epsilon^{-2})$-coloring can be computed deterministically in $O(\log^* n)$ rounds in the CONGEST model.
}
\begin{document}

\maketitle
\begin{abstract}
It is a well known fact that sequential algorithms which exhibit a strong "local" nature can be adapted to the distributed setting given a legal graph coloring. The running time of the distributed algorithm will then be at least the number of colors. Surprisingly, this well known idea was never formally stated as a unified framework. In this paper we aim to define a robust family of local sequential algorithms which can be easily adapted to the distributed setting. We then develop new tools to further enhance these algorithms, achieving state of the art results for fundamental problems.

We define a simple class of greedy-like algorithms which we call \emph{orderless-local} algorithms.
We show that given a legal $c$-coloring of the graph, every algorithm in this family can be converted into a distributed algorithm running in $O(c)$ communication rounds in the CONGEST model.
We show that this family is indeed robust as both the method of conditional expectations and the unconstrained submodular maximization algorithm of Buchbinder \etal \cite{BuchbinderFNS15} can be expressed as orderless-local algorithms for \emph{local utility functions} --- Utility functions which have a strong local nature to them.

We use the above algorithms as a base for new distributed approximation algorithms for the weighted variants of some fundamental problems: Max $k$-Cut, Max-DiCut, Max 2-SAT and correlation clustering. We develop algorithms which have the same approximation guarantees as their sequential counterparts, up to a constant additive $\epsilon$ factor, while achieving an $O(\log^* n)$ running time for deterministic algorithms and $O(\epsilon^{-1})$ running time for randomized ones. This improves exponentially upon the currently best known algorithms.

\end{abstract}


\section{Introduction}
A large part of research in the distributed environment aims to develop fast distributed algorithms for problems which have already been studied in the sequential setting. Ideally, we would like to use the power of the distributed environment to achieve a substantial improvement in the running time over the sequential algorithm, and indeed, for many problems distributed algorithms achieve an exponential improvement over the sequential case.
One approach to designing distributed algorithms is using the sequential algorithm as natural staring point \cite{Censor-HillelLS17,Bar-YehudaCS17,Bar-YehudaCGS17,GallagerHS83,BaswanaS07}, then certain adjustments are made for the distributed environment in order to achieve a faster running time.

There is a well known folklore in distributed computing, which roughly says that if a sequential graph algorithm works by traversing nodes in any order (perhaps adversarial), and for every node makes a local decision, then given a legal $c$-coloring of the graph, the algorithm can be adapted to the distributed setting by going over all color classes, and for each executing all nodes in the class simultaneously. Surprisingly, there is no formal framework describing the above. In this paper we provide such a framework for a specific class of algorithms (defined later). 

We note that for general graphs a legal coloring may require at least $\Delta+1$ colors, where $\Delta$ is the maximal degree of the graph. Using the above framework we aim to answer the following question: Are there certain classes of algorithms where using the above can result in a running time sublinear in $\Delta$? We show that for certain approximation problems the answer is quite surprising, as we are able to achieve an almost constant running time! 

More precisely, we show that for the problems of Max $k$-Cut, Max-DiCut, Max 2-SAT and correlation clustering we can adapt the sequential algorithm to these problems in such a way that the running time is $O(\log^* n)$ rounds for deterministic algorithms and $O(\epsilon^{2})$ for randomized ones, while losing only an additive $\epsilon$-factor in the approximation ratio. For the problems of Max-Cut and Max-DiCut this greatly improves upon the previous best known results, which required a number of rounds linear in $\Delta$. A summary of our results appears in Table~\ref{tab: results}.

\begin{table}[htbp]
	\label{tbl: results}
	\centering
	\begin{tabular*}{\linewidth}{@{}l@{\extracolsep{\fill}}c@{}c@{}c@{}c@{}r@{}}
		\toprule
		Problem       &  Our Approx.   & Our Time &  Prev Approx.   & Prev Time & Notes \\
		\midrule
		
		Weighted Correlation-Clustering*    &
		$1/2-\epsilon$&
		$O(\log^* n)$ &
		-&
		- &
		det.\\

		Weighted Max $k$-Cut    &
		$1-1/k-\epsilon$&
		$O(\log^* n)$ &
		$1/2$\cite{Censor-HillelLS17}\textsuperscript{$\dagger$}&
		$\tilde{O}(\Delta + \log^* n)$ &
		det.\\

		Weighted Max-Dicut    &
		$1/3-\epsilon$&
		$O( \log^* n)$  &
		$1/3$\cite{Censor-HillelLS17}\textsuperscript{$\dagger$}&
		$\tilde{O}(\Delta + \log^* n)$ &
		det.\\
		
		Weighted Max-Dicut    &
		$1/2-\epsilon$&
		$O(\epsilon^{-1})$  &
		$1/2$ \cite{Censor-HillelLS17}\textsuperscript{$\dagger$}&
		$\tilde{O}(\Delta + \log^* n)$ &
		rand.\\

		Weighted Max 2-SAT    &
		$3/4-\epsilon$&
		$O(\epsilon^{-1})$  &
		-&
		- &
		rand.\\

		\bottomrule
	\end{tabular*}
	\caption{Summary of our results for the $\mathsf{CONGEST}$ model ($\tilde{O}$ hides factors polylogarithmic in $\Delta$).
		(*) General graphs, max-agree
		($\dagger$) Unweighted graphs, only Max-Cut ($k=2$).
	}
	\label{tab: results}
\end{table}


\subsection{Tools and results}
In this paper we focus our attention on approximation algorithms for unconstrained optimization problems on graphs. We are given some graph $G(V,E)$, where each vertex $v$ is assigned a variable $X_v$ taking values in some set $A$.
We aim to maximize some utility function $f$ over these variables (For a formal definition see Section~\ref{sec: preliminaries}). Our distributed model is the CONGEST model of distributed computation, where the network is represented by a graph, s.t nodes are computational units and edges are communication links. Nodes communicate in synchronous communication rounds, where at each round a node sends and receives messages from all of its neighbors.
In the CONGEST model the size of messages sent between nodes is limited to $O(\log n)$ bits, where $\size{V}=n$. This is more restrictive than the LOCAL model, where message size is unbounded. Our complexity measure is the number of communication rounds of the algorithm.

Adapting a sequential algorithm of the type we describe above to the distributed setting, means we wish each node $v$ in the communication graph to output an assignment to $X_v$ such that the approximation guarantee is close to that of the sequential algorithm, while minimizing the number of communication rounds of the distributed algorithm.
Our goal is to formally define a family of sequential algorithms which can be easily converted to distributed algorithms, and then develop tools to allow these algorithms to run exponentially faster, while achieving almost the same approximation ratio. To achieve this we focus our attention on a family of sequential algorithms which exhibit a very strong local nature.



We define a family of utility functions, which we call \emph{local utility functions} (Formally defined in Section~\ref{sec: preliminaries}). We say that a utility function $f$ is a local utility function, if the change to the value of the function upon setting one variable $X_v$ can be computed locally.
Intuitively, while optimizing a general utility function in the distributed setting might be difficult for global functions, the local nature of the family of local utility functions makes it a perfect candidate.



We focus on adapting a large family of, potentially randomized, local algorithms to the distributed setting. We consider \emph{orderless-local} algorithms - algorithms that can traverse the variables in any order and in each iteration apply some local function to decide the value of the variable. By local we mean that the decision only depends on the local environment of the node in the graph, the variables of nodes adjacent to that variable and some randomness only used by that node. This is similar to the family of Priority algorithms first defined in \cite{BorodinNR02}. The goal of \cite{BorodinNR02} was to formally define the notion of a greedy algorithm, and then to explore the limits of these algorithms. Our definition is similar (and can be expressed as a special case of priority algorithms), but the goal is different. While \cite{BorodinNR02} aims to prove lower bounds, we provide some sufficient conditions that allow us to easily transform local sequential algorithms into fast distributed algorithms.

Our definitions are also similar to the SLOCAL model \cite{GhaffariKM17}, which also shows that sequential algorithms which traverse the graph vertices in any order and make local decisions can be adapted to the distributed LOCAL model in poly logarithmic rounds using randomization. While the results of \cite{GhaffariKM17} are much more broad, our transformation does not require any randomization and works in the CONGEST model. Finally, we should also mention the field of local computation algorithms \cite{RubinfeldTVX11} whose aim is developing efficient local \emph{sequential} algorithms. We refer the reader to an excellent survey by Levi and Medina \cite{levi2017centralized}.

One might expect that due to the locality of this family of algorithms it can be distributed if the graph is provided with a legal coloring. The distributed algorithm goes over the color classes one after another and executes all nodes in the color class simultaneously. This solves any conflicts that may occur form executing two neighboring nodes, while the orderless property guarantees that this execution is valid. In a sense this argument was already used for specific algorithm (Coloring to MIS \cite{Linial92}, MaxIS of \cite{Bar-YehudaCGS17}, Max-Cut of \cite{Censor-HillelLS17}). We provide a more general result, using this classical argument. Specifically, we show that given a legal $c$-coloring, any orderless-local algorithm can be distributed in $O(c)$-communication rounds in the CONGEST model.

To show that this definition is indeed robust, we show two general applications. The first is adapting the method of conditional expectations (Formally defined in Section~\ref{sec: preliminaries}) to the distributed setting. This method is inherently sequential, but we show that if the utility function optimized is a local utility function, then the algorithm is an orderless-local algorithm. A classical application of this technique is for Max $k$-cut, where an $(1-1/k)$-approximation is achieved when every node chooses a cut side at random. This can be derandomized using the method of conditional expectations, and adapted to the distributed setting, as the cut function is a local utility function. We note that the same exact approach results in a $(1/2-\epsilon)$-approximation for max-agree correlation clustering on general graphs (see Section~\ref{sec: preliminaries} for a definition). Because the tools used for Max-Cut directly translate to correlation clustering, we focus on Max-Cut for the rest of the paper, and only mention correlation clustering at the very end.

The second application is the unconstrained submodular maximization algorithms of \cite{BuchbinderFNS15}, where a deterministic 1/3-approximation and a randomized expected 1/2-approximation algorithms are presented. We show that both are orderless-local algorithms when provided with a local utility function. This can be applied to the problem of Max-DiCut, as it is an unconstrained submodular function, and also a local utility function. The algorithms of \cite{BuchbinderFNS15} were already adapted to the distributed setting for the specific problem of Max-DiCut by \cite{Censor-HillelLS17} using similar ideas. The main benefit of our definition is the convenience and generality of adapting these algorithms without the need to consider their analysis or reprove correctness.
We conclude that the family of orderless-local algorithms indeed contains robust algorithms for fundamental problems, and especially the method of conditional expectations. 

At the time this paper was first made public, there was no distributed equivalent for the method of conditional expectations. We have since learned that, independently and simultaneously, an adaptation of the method of conditional expectations to the distributed setting was also presented in \cite{MohsenDerand}. Their results show how the method of conditional expectations combined with a \emph{legal} coloring can be used to convert any randomized LOCAL $r$-round algorithm for a locally checkable problem to a deterministic one, running in $O(\Delta^{O(r)} + O(r\log^* n))$.\footnote{They actually show that the running time is either $O(\Delta^{O(r)} + O(r\log^* n))$ or $r\cdot2^{O(\sqrt{log n})}$, achieving the latter via network decomposition. We focus on the first bound, as the second is less relevant for the comparison which follows.} This is done via a transformation to an SLOCAL algorithm, where the derandomization is applied and then transforming back to a LOCAL algorithm. 

Although not stated for the CONGEST model, we believe it to be the case that when $r=1$ their application of the method of conditional expectations works in the CONGEST, and is equivalent to our results. Another difference apart from the different model of communication, is that they focus on derandomizing locally checkable problems, while we focus on local utility functions. These two families of problems are different, as the approximation guaranteed for a certain local utility function need not be locally checkable. This last point highlights the different goal of the two papers. While \cite{MohsenDerand} skillfully show that a large family of LOCAL algorithm can be derandomized, we aim to adapt sequential algorithm to the distributed setting while achieving as \emph{fast} of a running time as possible in the more restrictive CONGEST model -- hence we focus on local utility function which capture the locality of the optimization process.

Next, we wish to consider the running time of these algorithms. Recall that we expressed the running time of orderless local algorithms in terms of the colors of some legal coloring for the graph. For a general graph, we cannot hope for a legal coloring using less than $\Delta+1$, where $\Delta$ is the maximum degree in the graph. This means that using the distributed version of an orderless-local algorithm unchanged will have a running time linear in $\Delta$. We show how to overcome this obstacle for Max $k$-Cut and Max-DiCut. The general idea is to compute a \emph{defective} coloring of the graph which uses few colors, drop all monochromatic edges, and call the algorithm for the new graph which now has a legal coloring.

A key tool in our algorithms is a new type of defective coloring we call a \emph{weighted $\epsilon$-defective coloring}. The classical defective coloring allows each vertex to have at most $d$ monochromatic edges, for some defect parameter $d$. We consider positively edge weighted graphs and require a weighted fractional defect - for every vertex the total weight of monochromatic edges is at most an $\epsilon$-fraction of the total weight of all edges of that vertex. We show that a weighted $\epsilon$-defective coloring using $O(\epsilon^{-2})$ colors can be computed \emph{deterministically} in $O(\log^*n)$ rounds using the defective coloring algorithm of \cite{Kuhn09}. The classical algorithm of Kuhn was found useful in the adaptation of sequential algorithms to the distributed setting \cite{FischerGK17, GhaffariKM17}, thus its effectiveness for weighted $\epsilon$-defective coloring might be of further use.

Although we cannot guarantee a legal coloring with a small number of colors for any graph $G(V,E,w)$, we may remove some subset of $E$ which will result in a new graph $G'$ with a low chromatic number. We wish to do so while not decreasing the total sum of edge weights in $G'$, which we prove guarantees the approximation will only be mildly affected for our cut problems.
Formally, we show that if we only decrease the total edge weight by an $\epsilon$-fraction, we will incur an additive $\epsilon$-loss in the approximation ratio of the cut algorithms for $G$.
For the randomized algorithm this is easy, simply color each vertex randomly with a color in $[\ceil{\epsilon^{-1}}]$ and drop all monochromatic edges. For the deterministic case, we execute our weighted $\epsilon$-defective coloring algorithm, and then remove all monochromatic edges. We then execute the relevant cut algorithm on the resulting graph $G'$ which now has a legal coloring, using a small number of colors. The above results in extremely fast approximation algorithms for weighted Max $k$-Cut and weighted Max-DiCut, while having almost the same approximation ratio as their sequential counterpart.

Finally, our techniques can also be applied to the problem of weighted Max 2-SAT. To do so we may use the randomized expected 3/4-approximation algorithm presented in \cite{PoloczekSWZ17}. It is based on the algorithm of \cite{BuchbinderFNS15}, and thus is almost identical to the unconstrained submodular maximization algorithm. Because the techniques we use are very similar to the above, we defer the entire proof to the appendix.

\subsection{Previous research}
\textbf{Cut problems:} An excellent overview of the Max-Cut and Max-DiCut problems appears in \cite{Censor-HillelLS17}, which we follow in this section. Computing Max-Cut exactly is NP-hard as shown by Karp \cite{Karp72} for the weighted version, and by \cite{GareyJS76} for the unweighted case. As for approximations, it is impossible to improve upon a 16/17-approximation for Max-Cut and a 12/13-approximation for Max-DiCut unless $P=NP$ \cite{TrevisanSSW00, Hastad01}. If every node chooses a cut side randomly, an expected 1/2-approximation for Max-Cut, a 1/4-approximation for Max-DiCut and a $(1-1/k)$-approximation is achieved. This can be derandomized using the method of conditional expectations. In the breakthrough paper of Goemans and Williamson \cite{GoemansW95} a 0.878-approximation is achieved using semidefinite programming. This is optimal under the unique games conjecture \cite{KhotKMO07}. In the same paper a 0.796-approximation for Max-DiCut was presented. This was later improved to 0.863 in [MatuuraM01]. Other results using different techniques are presented in \cite{KaleS11, Trevisan12}.

In the distributed setting the problem has not received much attention. A node may choose a cut side at random, achieving the same guarantees as above in constant time. In \cite{HirvonenRSS14} a distributed algorithm for $d$-regular triangle free graphs which achieves a $(1/2 + 0.28125/\sqrt{d})$-approximation ratio in a single communication round is presented. The only results for  general graphs in the distributed setting is due to \cite{Censor-HillelLS17}. In the CONGEST model they present a deterministic 1/2-approximation for Max-Cut, a deterministic 1/3-approximation for Max-DiCut, and a randomized expected 1/2 approximation for Max-DiCut running in $\tilde{O}(\Delta + \log^* n)$ communication rounds. The results for Max-DiCut follow from adapting the unconstrained submodular maximization algorithm of \cite{BuchbinderFNS15} to the distributed setting. Better results are presented for the LOCAL model; we refer the reader to \cite{Censor-HillelLS17} for the full details.

\noindent \textbf{Max 2-SAT:}
The decision version of Max 2-SAT is NP-complete~\cite{GareyJS76}, and there exist several approximation algorithms~\cite{GoemansW95,FeigeG95,LewinLZ02,MatuuraM01}, of which currently the best known approximation ratio is 0.9401 \cite{LewinLZ02}. In \cite{Austrin07} it is shown that assuming the unique games conjecture, the approximation factor of \cite{LewinLZ02} cannot be improved. Assuming only that $P\neq NP$ it cannot be approximated to within a 21/22-factor \cite{Hastad01}. To the best of our knowledge the problem of Max 2-SAT (or Max-SAT) was not studied in the distributed model.

\noindent \textbf{Correlation clustering:}
An excellent overview of correlation clustering (see Section~\ref{sec: preliminaries} for a definition) appears in \cite{AhnCGMW15}, which we follow in this section.
Correlation clustering was first defined by \cite{BansalBC02}. Solving the problem exactly is NP-Hard, thus we are left with designing approximation algorithms for the problem, here one can try to approximate \emph{max-agree} or \emph{min-disagree}. If the graph is a clique, there exists a PTAS for max-agree \cite{BansalBC02,GiotisG06}, and a 2.06-approximation for max-disagree \cite{ChawlaMSY15}. For general (even weighted) graphs there exists a 0.7666-approximation for max-agree \cite{CharikarGW05, Swamy04}, and a $O(\log n)$-approximation for min-disagree \cite{DemaineEFI06}. A trivial 1/2-approximation for max-agree on general graphs can be achieved by considering putting every node in a separate cluster, then considering putting all nodes in a single cluster, and taking the more profitable of the two.

In the distributed setting little is known about correlation clustering. In \cite{Censor-HillelHK16} a dynamic distributed MIS algorithm is provided, it is stated that this achieves a 3-approximation for min-disagree correlation clustering as it simulates the classical algorithm of Ailon \etal \cite{AilonCN08}. We note that the algorithm of Ailon \etal assumes the graph to be a clique, thus the above result is limited to complete graphs where the edges of the communication graph are taken to be the positive edges, and the non-edges are taken as the negative edges (as indeed for general graphs, the problem is APX-Hard, and difficult to approximate better than $\Theta(\log n)$ \cite{DemaineEFI06}). We also note that using only two clusters, where each node chooses a cluster at random, guarantees an expected 1/2-approximation for max-agree on weighted general graphs. We derandomize this approach in this paper.

\section{Preliminaries}
\label{sec: preliminaries}
\noindent \textbf{Sequential algorithms:}
The main goal of this paper is converting (local) sequential graph algorithms for unconstrained maximization (or minimization) to distributed graph algorithms. Let us first define formally this family of algorithms. The sequential algorithm receives as input a graph $G=(V,E)$, we associate each vertex $v\in V$ with a variable $X_v$ taking values in some finite set $A$. The algorithm outputs a set of assignments $\Xbar=\set{X_v = \alpha_v}$. The goal of the algorithms is to maximize some utility function $f(G,\Xbar)$ taking in a graph and the set of assignments and outputting some value in $\mathbb{R}$. For simplicity we assume that the order of the variables in $\Xbar$ does not affect $f$, so we use a set notation instead of a vector notation. We somewhat abuse notation, and when assigning a variable we write $\Xbar \cup \set{X_v=\alpha}$, meaning that any other assignment to $X_v$ is removed from the set $\Xbar$. We also omit $G$ as a parameter when it is clear from context.

When considering randomized algorithms we assume the algorithm takes in a vector of random bits denoted by $\vec{r}$. This way of representing random algorithms is identical to having the algorithm generate random coins, and we use these two definitions interchangeably. The randomized algorithm aims to maximize the expectation of $f$, where the expectation is taken over the random bits of the algorithm.

\noindent \textbf{Max $k$-Cut, Max-DiCut: }
In this paper we provide fast distributed approximation algorithms to some fundamental problems, which we now define formally. In the Max $k$-Cut problem we wish to divide the vertices into $k$ disjoint sets, such that the weight of edges between different sets is maximized. In the Max-DiCut problem the edges are directed and we wish to divide the edges into two disjoint sets, denoted $A,B$, such that the weight of edges directed from $A$ to $B$ is maximized.

\noindent \textbf{Max 2-SAT: } In the Max 2-SAT problem we are given a set of unique weighted clauses over some set of variables, where each clause contains \emph{at most} two literals. Our goal is to maximize the weight of satisfied clauses. This problem is more general than the cut problems, so we must define what it means in the distributed context. First, the variables will be node variables as defined before. Second, each node knows all of the clauses it appears in as a literal.

\noindent \textbf{Correlation clustering: }
We are given an edge weighted graph $G(V,E,w)$, such that each edge is also assigned a value from $\set{+,-}$ (referred to positive and negative edges).
Given some partition, $C$, of the graph into disjoint clusters, we say that an edge \emph{agrees} with $C$ if it is positive and both endpoints are in the same cluster, or it is negative, and its endpoints are in different clusters. Otherwise we say it \emph{disagrees} with $C$. We aim to find a partition $C$, using any number of clusters, such that the weight of edges that agree with $C$ (agreements) is maximized (max-agree), or equivalently the weight of edges that disagree with $C$ is minimized (min-disagree).

The problem is usually expressed as an LP using edge variables, where each variable indicates whether the nodes are in the same cluster. This allows a solution to use any number of clusters. In this paper we only aim to achieve a $(1/2 - \epsilon)$-approximation for the problem. This can be done rather simply without employing the full power of correlation clustering. Specifically, two clusters are enough for our case as we show that we can deterministically achieve $(1/2-\epsilon)\size{E}$ agreements which results in the desired approximation ratio.

\noindent \textbf{Local utility functions: }
We are interested in a type of utility function which we call a \emph{local utility function}. Before we continue with the definition let us define an operator on assignments $\Xbar$, we define $L_v[\Xbar] = \set{\set{X_u=\alpha_u} \in \Xbar \mid u\in N(v)}$. For convenience, when we pass $\Xnv$ as parameter to a function, we assume that the function also receives the 1-hop neighborhood of $v$ which we do not write explicitly. We say that a utility function $f$, as defined above, is a local utility function if for every $v$ there exists a function $g_v$ s.t $f(\Xbar\cup \set{X_v =\alpha}) - f(\Xbar\cup \set{X_v =\alpha'}) = g_v(L_v[X], \alpha, \alpha')$.
That is, to compute the change in the utility function which is caused by changing $X_v$ from $\alpha'$ to $\alpha$, we only need to know the immediate neighborhood of $v$, and the assignment to neighboring node variables. We note that for the cut problems considered in this paper the utility functions are indeed local utility functions. This is proven in the following Lemma:

\begin{lemma}
\label{lem: cuts local}
\LemCutsLocal
\end{lemma}
\begin{proof}
	The utility functions for Max $k$-Cut is given by $f(\Xbar) = \sum_{e=(w,u)\in E} w(e) \cdot X_w \oplus X_u$ where $X_w \oplus X_u = 0$ if $X_w = X_u$ and 1 otherwise. Thus, if we fix some $v$ it holds that
	\begin{align*}
	&f(\Xbar \cup \set{X_v = \alpha'}) - f(\Xbar \cup \set{X_v = \alpha})
	\\ &= \sum_{e=(v,u)\in E} w(e) \cdot \alpha' \oplus X_u - \sum_{e=(v,u)\in E} w(e) \cdot \alpha \oplus X_u
	\\&=  \sum_{e=(v,u)\in E} w(e) \cdot ( \alpha' \oplus X_u - \alpha \oplus X_u) \triangleq g_v(L_v[\Xbar], \alpha', \alpha)
	\end{align*}
	
	Because the final sum only depends on vertices $u\in N(v)$, the last equality defines the local function equivalent to the difference, and we are done.
	
	For the problem of Max-DiCut the utility functions is given by $f(\Xbar) = \sum_{e=(v\rightarrow u)\in E} w(e) \cdot X_v \wedge (1-X_u)$, and for max-agree correlation clustering with 2 clusters the utility function is given by $f(\Xbar) = \sum_{e=(v,u)\in E^+ }  w(e) \cdot(1-X_v \oplus X_u) + \sum_{e=(v,u)\in E^- }  w(e) \cdot X_v \oplus X_u$ ($E^+, E^-$ are the positive and negative edges, respectively), and the proof is exactly the same.
\end{proof}

\noindent \textbf{Submodular functions: } 
A family of functions that will be of interest in this paper is the family of \emph{submodular functions}. A function $f:{0,1}^\Omega \rightarrow \mathbb{R}$ is called a set function, with ground set $\Omega$. It is said to be submodular if for every $S,T \subseteq \Omega$ it holds that $f(S) + f(T) \geq f(S\cup T) + f(S \cap T)$. The functions we are interested in have $V$ as their ground set, thus we remain with our original notation, setting $A=\set{0,1}$ and having $f$ take in a set of binary assignments $\Xbar$ as a parameter.


\noindent \textbf{The method of conditional expectations: }
Next, we consider the method of conditional expectations. Let $A$ be some set and $f: A^n \rightarrow \mathbb{R}$, next let $\Xbar = (X_1,...,X_n)$ be a vector of random variables taking values in $A$. We wish to be consistent with the previous notation, thus we treat $\Xbar$ as a set of assignments.
If $E[f(\Xbar)] \geq \beta$, then there is an assignment of values $\overline{Z}=\set{X_i=\alpha_i}_{i=1}^n$ such that $f(\overline{Z}) \geq \beta$. We describe how to \emph{find} the vector $Z$.
We first note that from the law of total expectation it holds that
$E[f(\Xbar)] = \sum_{\alpha \in A} E[f(\Xbar) \mid X_1 = \alpha]Pr[X_1=\alpha]$, and therefore for at least some $\alpha \in A$ it holds that $E[f(\Xbar) \mid X_1 = \alpha] \geq \beta$. We set this value to be $\alpha_1$. We then repeat this process for the rest of the values in $\Xbar$, which results in the set $\overline{Z}$. In order for this method to work we need it to be possible to \emph{compute}\footnote{This point is critical, and this computation is not simple in many cases. In our case we also need this computation to be done \emph{locally} at every nodes. We apply this technique to Max-Cut, which meets all of these demands.} the conditional expectation of $f(\Xbar)$.


\noindent \textbf{Graph coloring: }
A $c$-coloring for $G(V,E)$ is defined as a function $\varphi:V \rightarrow \mathcal{C}$. For simplicity we treat any set $\mathcal{C}$ of size $c$ with some ordering as the set of integers $[c]$. This simplifies things as we can always consider $\varphi(v) \pm 1$, which is very convenient.
We say that a coloring is a legal coloring if $\forall v,u$ s.t $(v,u)\in E$ it holds that $\varphi(v) \neq \varphi(u)$.
An important tool in this paper is \emph{defective coloring}. Let us fix some $c$-coloring function $\varphi:V \rightarrow [c]$.
We define the defect of a vertex to be the number of monochromatic edges it has. Formally, $defect(v)=size{\set{u\in N(v) \mid \varphi(v)=\varphi(u)}}$. We call $\varphi$ a $c$-coloring with defect $d$ if it holds that $\forall v\in V, defect(v) \leq d$. A classic result by Kuhn \cite{Kuhn09} states that for all $d \in \set{1,2,...,\Delta}$ an $O(\Delta ^2 / d^2)$-coloring with defect $d$ can be computed deterministically in $O(\log^* n)$ rounds in the CONGEST model.

In this paper we define a new kind of defective coloring which we call a \emph{weighted $\epsilon$-defective coloring}. Given a positively edge weighted graph and any coloring, for every vertex we denote by $E_m(v)$ its monochromatic edges. Define its weighted defect as $defect_w(v) = \sum_{e=(u,v)\in E_m(v)} w(e)$. We aim to find a coloring s.t the defect for every $v$ is below $\epsilon w(v) = \epsilon \sum_{v\in e} w(e)$.
We show that the algorithm of Kuhn actually computes a weighted $\epsilon$-defective $O(\epsilon^{-2})$-coloring. We state the following theorem (As the analysis is rather similar to the original analysis of Kuhn, the proof is deferred to the appendix):

\begin{theorem}
\label{thm: eps-coloring}
\ThmKuhnExt
\end{theorem}

\section{Orderless-local algorithms}

Next we turn our attention to a large family of (potentially randomized) greedy algorithms. We limit ourselves to graph algorithms s.t every node $v$ has a variable $X_v$ taking values in some set $A$. We aim to maximize some global utility function $f(\Xbar)$. We focus on a class of algorithms we call \emph{orderless-local} algorithms. These are greedy-like algorithms which may traverse the vertices in any order, and at each step decide upon a value for $X_v$. This decision is local, meaning that it only depends on the 1-hop topology of $v$ and the values of neighboring variables. The decision may be random, but each variable has its own random bits, keeping the decision process local.

The code for a generic algorithm of this family is given in Algorithm~\ref{alg: orderless-local}. The algorithm first initiates the vertex variables.
Next it traverses the variables in some order $\pi:V \rightarrow [n]$. Each $X_{v_i}$ is assigned a value according to some function $decide$, which only depends on $L_v[\Xbar]$ at the time of the assignment and some random bits $\vec{r}_i$ which are only used to set the value for that variable. Finally the assignment to the variables is returned. We are guaranteed that the expected value of $f$ is at least $\beta(G)$ for any, potentially adversarial, ordering $\pi$ of the variables. Formally, $E_{\vec{r}}[f(OL(G, \vec{r}, \pi))] \geq \beta(G)$.

\RestyleAlgo{boxruled}
\LinesNumbered
\begin{algorithm}[htbp]
	
	\caption{\texttt{OL($G, \vec{r}, \pi$)}}
	\label{alg: orderless-local}
	$\forall v \in V, X_v = init(\Xnv)$ \\
	Order the variables according to $\pi$: $v_1, v_2...,v_n$\\
	\DontPrintSemicolon
	\lFor{$i$ from 1 to $n$ }
	{
		$X_{v_i} = decide(\Xnv, r_i)$
	}
	Return $\Xbar$\\
\end{algorithm}

We show that this family of algorithms can be easily distributed using coloring, s.t the running time of the distributed version depends on the number of colors. The distributed version, \texttt{OLDist}, is presented as Algorithm~\ref{alg: orderless-local-dist}. The variables are all initiated as in the sequential version, and then the color classes are executed sequentially, while in each color class the nodes execute $decide$ simultaneously, and send the newly assigned value to all neighbors. Decide does not communicate with the neighbors, so the algorithm finishes in $O(c)$ rounds.
\RestyleAlgo{boxruled}
\LinesNumbered
\begin{algorithm}[htbp]
	\caption{\texttt{OLDist($G, \vec{r}, \varphi$)}}
	\label{alg: orderless-local-dist}

	$\forall v \in V, X_v = init(\Xnv)$ \\
	\For{$i$ from 1 to $c$ }
	{
		\ForEach{$v$ s.t $\varphi(v)=i$ simultaneously}
		{
			$X_{v_i} = decide(\Xnv, r_i)$\\
			Send $X_{v_i}$ to neighbors\\
		}
	}
	return $\Xbar$\\
\end{algorithm}

It is easy to see that given the same randomness both the sequential and distributed algorithms output the same result, this is because all decisions of the distributed algorithm only depend on the 1-hop environment of a vertex, and we are provided with a \emph{legal} coloring. Thus, one round of the distributed algorithm is equivalent to many steps of the sequential algorithm. We prove the following lemma:

\begin{lemma}
\label{lem: exists ordering}
For any graph $G$ with a legal coloring $\varphi$, there exists an order $\pi$ on the variables s.t it holds that $OL(G,\vec{r}, \pi) = OLDist(G,\vec{r},\pi)$ for any $\vec{r}$.
\end{lemma}
\begin{proof}
	We prove the claim by induction on the executions of color classes by the distributed algorithm. We note that the execution of the distributed algorithm defines an order on the variables. Let us consider the $i$-th color class. Let us denote these variables as $\set{X_{v_j}}_{j=1}^k$, assigning some arbitrary order within the class. The ordering we analyze for the sequential algorithm would be $\pi(v_j) = (\varphi(v), j)$. Now both the distributed and sequential algorithms follow the same order of color classes, thus we allow ourselves to talk about the sequential algorithm finishing an execution of a color class.
	
	Let $Y_i$ be the assignments to all variables of the distributed algorithm after the $i$-th color class finishes execution. And let $Y'_i$ be the assignments made by the sequential algorithm following $\pi$ until all variable in the $i$-th color class are assigned. Both algorithms initiate the variables identically, so it holds that $Y_0 = Y'_0$. Assume that it holds that $Y_{i-1} = Y'_{i-1}$.  The coloring is legal, so for any $X_u, X_v$, s.t $\varphi(u)=\varphi(v)=i$ it holds that $N(v)\cap u = \emptyset$. Thus, when assigning $v$, its neighborhood is not affected by any other assignments done in the color class, so the randomness is identical for both algorithms, and using the induction hypothesis all assignments up until this color class were identical. Thus, for all variables in this color class $decide$ will be executed with the same parameters for both the distributed and sequential algorithms, and all assignments will be identical.
\end{proof}


Finally we show that for any graph $G$ with a legal coloring $\varphi$, it holds that $E_{\vec{r}}[f(OLDist(G, \vec{r}, \varphi))] \geq \beta(G)$. We know from Lemma~\ref{lem: exists ordering} that for any coloring $\varphi$ there exists an ordering $\pi$ s.t $OL(G,\vec{r}, \pi) = OLDist(G,\vec{r},\varphi)$ for any $\vec{r}$. The proof is direct from here: 
\begin{align*}
&E_{\vec{r}}[f(OLDist(G, \vec{r}, \varphi))] = \sum_{\vec{r}} Pr[\vec{r}] f(OLDist(G, \vec{r}, \varphi)) \\&= \sum_{\vec{r}} Pr[\vec{r}] f(OL(G, \vec{r}, \pi)) = E_{\vec{r}}[f(OL(G, \vec{r}, \pi))] \geq \beta(G)
\end{align*}

We conclude that any orderless-local algorithm can be distributed, achieving the same performance guarantee on $f$, and requiring $O(c)$ communication rounds to finish, given a legal $c$-coloring. We state the following theorem:

\begin{theorem}
\label{thm: orderless greedy}
Given some utility function $f$, any sequential orderless-local algorithm for which it holds that $E_{\vec{r}}[f(OL(G,\vec{r},\pi))] \geq \beta(G)$, can be converted into a distributed algorithm for which it holds that $E_{\vec{r}}[f(OLDist(G,\vec{r},\varphi))] \geq \beta(G)$, where $\varphi$ is a legal $c$-coloring of the graph. The running time of the distributed algorithm is $O(c)$ communication rounds.
\end{theorem}

\subsection{Distributed derandomization}
\label{sec: dist derand}
We consider the method of conditional expectations in the distributed case for some local utility function $f(G,\Xbar)$, as defined in the preliminaries. Assume that the value of every $X_v$ is set independently at random according to some distribution on $A$ which depends only on the 1-hop neighborhood of $v$. We are guaranteed that  $\forall G, E[f(G,\Xbar)] \geq \beta(G)$. Thus in the sequential setting we may use the method of conditional expectations to compute a deterministic assignment to the variables with the same guarantee. We show that because $f$ is a local utility function, the method of conditional expectations applied on $f$ is an orderless-local algorithm, and thus can be distributed.

Initially all variables are initiated to some value $\emptyset \notin A$, meaning the variable is unassigned.
Let $Y=\set{X_u = \alpha_u \mid u\in U\subseteq V}$ be some partial assignment to the variables. The method of conditional expectations goes over the variables in any order, and in each iteration sets $X_{v_i} = \argmax_{\alpha} E[f(\Xbar) \mid Y, X_{v_i}=\alpha]$. This is equivalent to $\argmax_{\alpha} \set{E[f(\Xbar) \mid Y, X_{v_i}=\alpha] - E[f(\Xbar) \mid Y]}$, as the subtracted term is just a constant. With this in mind, we present the pseudo code for the method of conditional expectations in Algorithm~\ref{alg: cond-exp}.

\RestyleAlgo{boxruled}
\LinesNumbered
\begin{algorithm}[htbp]	
	\caption{\texttt{CondExpSeq($G$)}}
	\label{alg: cond-exp}
	$\forall v \in V, X_v = \emptyset$ \\
	Order the variables according to any order: $v_1, v_2...,v_n$\\
	\DontPrintSemicolon
	\lFor{$i$ from 1 to $n$ }
	{
		$X_{v_i} = \argmax_{\alpha} E[f(\Xbar) \mid Y, X_{v_i} =\alpha_v] - E[f(\Xbar) \mid Y]$
	}
	Return $\Xbar$\\
\end{algorithm}

To show that Algorithm~\ref{alg: cond-exp} is an orderless-local algorithm we only need to show that $\argmax_{\alpha} E[f(\Xbar) \mid Y, X_v =\alpha_v] - E[f(\Xbar) \mid Y]$ can be computed locally for any $v$. We state the following lemma, followed by the main theorem for this section.

\begin{lemma}
\label{lem: cond-exp local}
The value $\argmax_{\alpha} E[f(\Xbar) \mid Y, X_v =\alpha_v] - E[f(\Xbar) \mid Y]$ can be computed locally.
\end{lemma}
\begin{proof}
	It holds that:
	\begin{align*}
	&E[f(\Xbar) \mid Y, X_v =\alpha_v] - E[f(\Xbar) \mid Y]
	\\&= \sum_{\alpha \in A} E[f(\Xbar) \mid Y, X_v =\alpha_v]Pr[X_v = \alpha] - \sum_{\alpha \in A}E[f(\Xbar) \mid Y,X_v = \alpha]Pr[X_v=\alpha]
	\\&= \sum_{\alpha \in A} Pr[X_v=\alpha] (E[f(\Xbar) \mid Y, X_v =\alpha_v] - E[f(\Xbar) \mid Y,X_v = \alpha])
	\end{align*}
	Where the first equality is due to the law of total expectation and the fact that $\sum_{\alpha \in A} Pr[X_v=\alpha] =1$. The probability of assigning $X_v$ to some value can be computed locally, so we are only left with the difference between the expectations.
	To show that this is indeed a local quantity we use the definition of expectation as a weighted summation over all possible assignments to unassigned variables.
	Let $U_v$ be the set of all possible assignments to unassigned variables in $N(v)$ and let $U$ be the set of all possible assignments to the rest of the unassigned variables.
	It holds that:
	\begin{align*}
	&E[f(\Xbar) \mid Y, X_v =\alpha_v] - E[f(\Xbar) \mid Y,X_v = \alpha]
	\\&=\sum_{Z_v \in U_v}\sum_{Z\in U} Pr[Z_v]Pr[Z] f(\Xbar \cup Z_v \cup Z \cup \set{X_v = \alpha_v}) - f(\Xbar \cup Z_v \cup Z \cup \set{X_v = \alpha})
	\\&=\sum_{Z_v\in U_v}\sum_{Z \in U} Pr[Z_v]Pr[Z] g_v(L_v[\Xbar \cup Z_v \cup Z], \alpha, \alpha_v)
	=\sum_{Z_v\in U_v}\sum_{Z \in U} Pr[Z_v]Pr[Z] g_v(L_v[\Xbar \cup Z_v], \alpha, \alpha_v)
	\\&=\sum_{Z_v\in U_v} Pr[Z_v] g_v(L_v[\Xbar \cup Z_v], \alpha, \alpha_v), 
	\end{align*}
	
	where in the first equality we use the definition of expectations and the fact that the variables are set independently of each other. Then we use the definition of a local utility function, and finally the dependence on $U$ disappears due to the law of total probability. The final sum can be computed locally, as the probabilities for assigning variables in $Z_v$ are known and $g_v$ is local.
\end{proof}

\begin{theorem}
\label{thm: derand color}
Let $G$ be any graph and $f$ a local utility function for which it holds that $E[f(\Xbar)] \geq \beta$, where the random assignments to the variables are independent of each other, and depend only on the immediate neighborhood of the node. There exists a distributed algorithm achieving the same value as the expected value for $f$, running in $O(c)$ communication rounds in the CONGEST model, given a legal $c$-coloring.
\end{theorem}



\subsection{Submodular Maximization}
In this section we consider the problem of unconstrained submodular function maximization. Given an submodular function $f$ (as defined in Section~\ref{sec: preliminaries}), we aim to find an input s.t the function is maximized. 
There are no constraints on the input set we pass to the function, hence it is 'unconstrained'. 
We are interested in finding an approximate solution to the problem, to this end, we consider both the deterministic and randomized algorithms of \cite{BuchbinderFNS15}, achieving 1/3 and 1/2 approximation ratios for unconstrained submodular maximization. We show that both can be expressed as orderless-local algorithms for any local utility function. 
As the deterministic and randomized algorithms of \cite{BuchbinderFNS15} are almost identical, we focus on the randomized algorithm achieving a 1/2-approximation in expectation (Algorithm~\ref{alg: rand-submodular}), as it is a bit more involved (The deterministic algorithm appears as Algorithm~\ref{alg: det-submodular}). The algorithms of \cite{BuchbinderFNS15} are defined for any submodular function, but as we are interested only in the case where the ground set is $V$, we will present it as such.

The algorithm maintains two variable assignment $Z_i, Y_i$, initially $Z_0 = \set{X_v = 0 \mid v\in V}$, $Y_0 = \set{X_v = 1 \mid v\in V}$. It iterates over the variables in any order, at each iteration it considers two nonnegative quantities $a_i,b_i$. These quantities represent the gain of either setting $X_{v_i}=1$ in $Z_{i-1}$ or setting $X_{v_i}=0$ in $Y_{i-1}$. Next a coin is flipped with probability $p = a_i / (a_i + b_i)$, if $a_i=b_i=0$ we set $p=1$. If we get heads we set $X_{v_i}=1$ in $Z_i$ and otherwise we set it to 0 in $Y_i$. When the algorithm ends it holds that $Z_n=Y_n$, and this is our solution. The deterministic algorithm is almost identical, only that it allows $a_i, b_i$ to take negative values, and instead of flipping a coin it makes the decision greedily by comparing $a_i, b_i$.


\RestyleAlgo{boxruled}
\LinesNumbered
\begin{algorithm}[htbp]
	\caption{\texttt{det-usm($f$)}}
	\label{alg: det-submodular}
	$Z_0 = \set{X_v = 0 \mid v\in V}$, $Y_0 = \set{X_v = 1 \mid v\in V}$\\
	\For{$i$ from 1 to $n$ }
	{
		$a_i = f(Z_{i-1} \cup \set{X_{v_i}=1}) - f(Z_{i-1})$\\
		$b_i = f(Y_{i-1} \cup \set{X_{v_i}=0}) - f(Y_{i-1})$\\
		
		\If{$a_i \geq b_i$}
		{
			$Z_i = Z_{i-1} \cup \set{X_{v_i}=1}$\\
			$Y_i = Y_{i-1} $\\
		}
		\Else
		{
			$Z_i = Z_{i-1}$\\
			$Y_i = Y_{i-1} \cup \set{X_{v_i}=0}$\\
		}
	}
	return $Z_n$\\
\end{algorithm}

\RestyleAlgo{boxruled}
\LinesNumbered
\begin{algorithm}[htbp]
	\caption{\texttt{rand-usm($f$)}}
	\label{alg: rand-submodular}

    $Z_0 = \set{X_v = 0 \mid v\in V}$, $Y_0 = \set{X_v = 1 \mid v\in V}$\\
    Order the variables in any order $v_1,...,v_n$\\
	\For{$i$ from 1 to $n$ }
	{
    	$a_i = \max \set{f(Z_{i-1} \cup \set{X_{v_i}=1}) - f(Z_{i-1}), 0}$\\
        $b_i = \max \set{f(Y_{i-1} \cup \set{X_{v_i}=0}) - f(Y_{i-1}), 0}$\\
        \DontPrintSemicolon
        \textbf{if} $a_i + b_i = 0$ \textbf{then} $p = 1$ \textbf{else} $p = a_i / (a_i + b_i)$\\
        $Y_i = Y_{i-1}, Z_i = Z_{i-1} $\\

        Flip a coin with probability $p$, \textbf{if} heads $Z_i = Z_{i} \cup \set{X_{v_i}=1}$, \textbf{else}  $Y_i = Y_{i} \cup \set{X_{v_i}=0}$\\
    }
    return $Z_n$\\
\end{algorithm}

We first note that the algorithm does not directly fit into our mold, as each vertex has two variables. We can overcome this, by taking $X_v$ to be a binary tuple, the first coordinate stores its value for $Z_i$, and the other for $Y_i$. Initially it holds that $\forall v\in V, X_v = (0,1)$, and our final goal function will only take the first coordinate of the variable.
We note that because $f$ is a local utility function the values $a_i,b_i$ can be computed locally, this results directly from the definition of a local utility function, as we are interested in the change in $f$ caused by flipping a single variable. Now we may rewrite the algorithm as an orderless-local algorithm, the pseudocode as Algorithm~\ref{alg: rand-submodular-OL}.

\RestyleAlgo{boxruled}
\LinesNumbered
\begin{algorithm}[htbp]
	\caption{\texttt{rand-usm($G,\vec{r}, \pi$)}}
	\label{alg: rand-submodular-OL}
	$\forall v\in V, X_v = (0,1)$\\
	Order the vertices according to $\pi$\\
	\For{$i$ from 1 to $n$ }
	{
		$X_u= decide(\Xnv,\vec{r}_i)$\\
	}
	return $X_n$\\
\end{algorithm}

\RestyleAlgo{boxruled}
\LinesNumbered
\begin{algorithm}[htbp]
	\caption{\texttt{decide($\Xnv, r$)}}
	\label{alg: rand-submodular-OL-decide}
	\DontPrintSemicolon
	
	$Z = \set{X_u = \alpha_{u,1} \mid \set{X_u = (\alpha_{u,1},\alpha_{u,2})} \in \Xnv}$\\
	
	$Y = \set{X_u = \alpha_{u,2} \mid \set{X_u = (\alpha_{u,1},\alpha_{u,2})} \in \Xnv}$\\
	
	$a = \max \set{g_v(Z,0,1), 0}$\\
	$b = \max \set{g_v(Z,1,0), 0}$\\

	\lIf{$a + b = 0$}
	{
		$p = 1$
	}
	\lElse
	{
		$p = a / (a + b)$
	}
	Flip coin with probability $p$\\
	\lIf{heads}
	{
		return (1,1)
	}
	\lElse
	{
		return (0,0)
	}
	
\end{algorithm}

Using Theorem~\ref{thm: orderless greedy} we state our main result:

\begin{theorem}
\label{thm: usm color}
For any graph $G$ and a local unconstrained submodular function $f$ with $V$ as its ground set, there exists a randomized distributed 1/2-approximation, and a deterministic 1/3-approximation algorithms running in $O(c)$ communication rounds in the CONGEST model, given a legal $c$-coloring.
\end{theorem}

\subsection{Fast approximations for cut functions}
\label{sec: fast cuts}
Using the results of the previous sections we can provide fast and simple approximation algorithms for Max-DiCut and Max $k$-Cut. Lemma~\ref{lem: cuts local} guarantees that the utility functions for these problems are indeed local utility functions. For Max-DiCut we use the algorithms of Buchbinder \etal, as this is an unconstrained submodular function. For Max $k$-Cut each node choosing a side uniformly at random achieves a $(1-1/k)$ approximation, thus we use the results of Section~\ref{sec: dist derand}.
Theorem~\ref{thm: usm color} and Theorem~\ref{thm: derand color} immediately guarantee distributed algorithms, running in $O(c)$ communication rounds given a legal $c$-coloring.


Denote by $Cut(G,\varphi)$ one of the cut algorithms guaranteed by Theorem~\ref{thm: usm color} or Theorem~\ref{thm: derand color}. We present two algorithms, \texttt{approxCutDet}, a deterministic algorithm to be used when $Cut(G,\varphi)$ is deterministic (Algorithm~\ref{alg: simple cut det}), and, \texttt{approxCutRand}, a randomized algorithm (Algorithm~\ref{alg: simple cut rand}) for the case when $Cut(G,\varphi)$ is randomized.
\texttt{approxCutDet} works by coloring the graph $G$ using a weighted $\epsilon$-defective coloring and then defining a new graph $G'$ by dropping all of the monochromatic edges. This means that the coloring is a legal coloring for $G'$. Finally we call one of the deterministic cut functions. \texttt{approxCutRand} is identical, apart from the fact that nodes choose a color uniformly at random from $[\ceil{\epsilon^{-1}}]$.

\RestyleAlgo{boxruled}
\LinesNumbered
\begin{algorithm}[htbp]
	\caption{\texttt{approxCutDet}$(G,\epsilon)$}
	\label{alg: simple cut det}
	$\varphi$ = epsilonColor$(G, \epsilon)$\\
	Let $G' = (V, E'=\set{(v,u)\in E \mid \varphi(v) \neq \varphi(u) })$\\
	Cut$(G', \varphi)$\\
\end{algorithm}

\RestyleAlgo{boxruled}
\LinesNumbered
\begin{algorithm}[htbp]
	\caption{\texttt{approxCutRand}$(G,\epsilon)$}
	\label{alg: simple cut rand}
	Each vertex $v$ chooses $\varphi(v)$ uniformly at random from $[\ceil{\epsilon^{-1}}]$\\
	Let $G' = (V, E'=\set{(v,u)\in E \mid \varphi(v) \neq \varphi(u) })$\\
	Cut$(G', \varphi)$\\
\end{algorithm}

For \texttt{approxCutDet}, the running time of the coloring is $O(\log^* n)$ rounds, returning a weighted $\epsilon$-defective $O(\epsilon^{-2})$-coloring. The running time of the cut algorithms is the number of colors, thus the total running time of the algorithm is $O(\epsilon^{-2}  + \log^* n)$ rounds. Using the same reasoning, the running time of \texttt{approxCutRand} is $O(\epsilon^{-1})$.
It is only left to prove the approximation ratio. We prove the following lemma:
\begin{lemma}
\label{lem: drop edges approx}
Let $G(V,E,w)$ be any graph, and let $G'(V,E',w)$ be a graph resulting from removing any subset of edges from $G$ of total weight at most $\epsilon \sum_{e\in E} w(e)$. Then for any constant $p$, any $p$-approximation for Max-DiCut or Max $k$-Cut for $G'$ is a $p(1-4\epsilon)$-approximation for $G$.
\end{lemma}
\begin{proof}
Let $OPT, OPT'$ be the size of optimal solutions for $G, G'$. It holds that $OPT' \geq OPT - \epsilon \sum_{e\in E} w(e)$, as any solution for $G$ is also a solution for $G'$ whose value differs by at most $\epsilon \sum_{e\in E} w(e)$ (the weight of discarded edges).
Assigning every node a cut side uniformly at random the expected cut weight is at least $\sum_{e\in E} w(e)/4$ for Max-DiCut and Max $k$-Cut. Using the probabilistic method this implies that $OPT \geq \sum_{e\in E} w(e)/4$. Using all of the above we can say that given a $p$-approximate solution for $OPT'$ it holds that:
$p\cdot OPT' \geq p(OPT - \epsilon \sum_{e\in E} w(e)) \geq p(OPT - 4\epsilon OPT) = p(1-4\epsilon)OPT $
\end{proof}
Lemma~\ref{lem: drop edges approx} immediately guarantees the approximation ratio for the deterministic algorithm. As for the randomized algorithm, let the random variable $\delta$ be the fraction of edges removed, let $p$ be the approximation ratio guaranteed by one of the cut algorithms and let $\rho$ be the approximation ratio achieved by \texttt{approxCutRand}. We know that $E_\rho[\rho \mid \delta] = p(1-4\delta)$. Applying the law of total expectations we get that $E[\rho] = E_\delta [E_\rho[\rho \mid \delta]]=E_\delta[p(1-4\delta)] = p(1-4\epsilon)$. We state our main theorems for this section.

\begin{theorem}
\label{thm: defective dicut}
There exists a deterministic $(1-1/k-\epsilon)$-approximation algorithms for Weighted Max $k$-Cut running in $O( \log^* n)$ communication rounds in the CONGEST model.
\end{theorem}

\begin{theorem}
\label{thm: defective dicut}
There exists a deterministic $(1/3-\epsilon)$-approximation algorithm for Weighted Max-DiCut running in $O(\log^* n)$ communication rounds in the CONGEST model.
\end{theorem}

\begin{theorem}
\label{thm: defective dicut}
There exists a randomized distributed expected $(1/2-\epsilon)$-approximation for Weighted Max-DiCut running in $O(\epsilon^{-1})$ communication rounds in the CONGEST model.
\end{theorem}

\paragraph{Correlation clustering} We note the same techniques used for Max-Cut work directly for max-agree correlation clustering on general graphs. Specifically, if we divide the nodes into two clusters, s.t each node selectes a cluster uniformly at random, each edge has exactly probability 1/2 to agree with the clustering, thus the expected value of the clustering is $\sum_{e\in E} w(e)/2$, which is a 1/2-approximation. The above can be derandomized exactly in the same manner as Max-Cut, meaning this is an orderless local algorithm. Finally, we apply the weighted $\epsilon$-defective coloring algorithm twice (note that we ignore the sign of the edge),
discard all monochromatic edges and execute the deterministic algorithm guaranteed from Theorem~\ref{thm: derand color} with a legal coloring. Because there must exists a clustering which has a value at least $\sum_{e\in E} w(e)/2$, a lemma identical to Lemma~\ref{lem: drop edges approx} can be proved and hence we are done. We state the following theorem:

\begin{theorem}
\label{thm: defective dicut}
There exists a deterministic $(1/2-\epsilon)$-approximation algorithms for weighted max-agree correlation clustering on general graphs, running in $O( \log^* n)$ communication rounds in the CONGEST model.
\end{theorem}


\bibliography{DMC}
\appendix
\newpage
\section{Extending Kuhn's algorithm for $\epsilon$-defective coloring}
\label{sec: kuhn}
We aim to extend the algorithm of Kuhn \cite{Kuhn09} to the case of a weighted $\epsilon$-defective coloring. Given a positively edge weighted graph and any coloring, for every vertex we denote by $E_m(v), E_b(v)$ its monochromatic and bi-chromatic edges respectively. Define its weighted defect as $defect_w(v) = \sum_{e=(u,v)\in E_m(v)} w(e)$. We aim to find a coloring s.t the defect for every $v$ is below $\epsilon w(v) = \epsilon \sum_{v\in e} w(e)$.

We show that Kuhn's algorithm can be adapted to this problem.


\paragraph{The algorithm (Algorithm 1)}: Given some coloring $[M]$ of the graph, each iteration of Kuhn's algorithm consists of assigning some unique function $\phi_x: A \rightarrow B$ to every color $x \in [M]$ (We abuse notation and also denote by $\phi_v$ the function assigned to $v$).
Then every node iterates over all $\alpha \in A$ and picks $\alpha$ s.t $ \sum_{v\in e, \phi_v(\alpha) = \phi_u(\alpha)} w(e) $ is minimized. We note that the family of functions used is the family of polynomials degree at most $k$ over some field.  Finally the vertex is assigned the color $(\alpha, \phi_v(\alpha))$.



\paragraph{Analysis:} We state a lemma analogous to Lemma 4.1 in \cite{Kuhn09}.
We show that it is possible to assign a color to $v$ s.t the defect of the vertex increases by at most $\epsilon \sum_{e \in E_b(v)} w(e)$.

\begin{lemma}
Assume we are given an $M$-coloring of $G$. For a value $k>0$ let the functions assigned to the colors be such that for any two colors $x,y \in [M]$, the functions $\phi_x, \phi_y$ intersect for at most $k$ values and s.t that $\size{A} > k \epsilon^{-1}$ then the increase in the weighted defect of every vertex $v$ of the new coloring computed by the algorithm is at most $\epsilon \sum_{e \in E_b(v)} w(e)$.
\end{lemma}

\begin{proof}
We note that in the proof of Lemma 4.1 in \cite{Kuhn09}, it is shown that if there two neighboring nodes $u,v$ with different colors that choose values $\alpha_u, \alpha_v$ s.t $\phi_v(\alpha_v) \neq \phi_u(\alpha_u)$ then after applying Algorithm 1, the edge will be bi-chromatic.

We show that for every node $v$ there exists some value $\alpha \in A$ s.t the increase in the weighted defect is at most $\epsilon \sum_{e \in E_b(v)} w(e)$.
Fix some node $v$ and assume by contradiction that this is not the case, thus for every $\alpha \in A$ it holds that the increase in weight of monochromatic edges is greater than $\epsilon \sum_{e \in E_b(v)} w(e)$. Let us count the total weight of monochromatic edges for all values of $\alpha$ in two different ways:
\begin{align*}
    k \sum_{e \in E_b(v)} w(e) \geq \sum_{e=(u,v) \in E_b(v)} \sum_{\alpha \in A} w(e) \cdot \chi_{\phi_v(\alpha) = \phi_u(\alpha)} \\=  \sum_{\alpha \in A} \sum_{e=(u,v) \in E_b(v)} w(e) \cdot \chi_{\phi_v(\alpha) = \phi_u(\alpha)} > \size{A} \epsilon \sum_{e \in E_b(v)} w(e)
\end{align*}
Thus, we get that $\size{A} < \epsilon^{-1} k$ which is a contradiction.

\end{proof}

Next the construction of such a family of functions remains the same using polynomials. Specifically, restating Theorem 4.6 in \cite{Kuhn09} using our parameters we get that:

\begin{theorem}
\label{thm: theorem 4.6}

Assume we are given an $M$ coloring of the graph. There are an explicit function $\phi_x$ for $x\in[M]$ and a constant $C_D>0$ s.t Algorithm 1 computes a $C_D\epsilon^{-2}\log^2_{1/\epsilon} M$ coloring s.t the weighted defect of every vertex $v$ increases by at most $\epsilon w(v)$.
\end{theorem}

Finally we prove a theorem equivalent to Theorem 4.9 in \cite{Kuhn09} using our parameters. The proof is almost identical to the original one, and we present it here for completeness.
\begin{theorem-repeat}{thm: eps-coloring}
\ThmKuhnExt
\end{theorem-repeat}
\begin{proof}
Let $M=[n]$ be the coloring induced by the unique identifiers of the vertices. Define $\epsilon_i = \epsilon_0 / 2^{T-i}, \epsilon_0 = \epsilon/2$, where $T$ is the smallest positive integer s.t $\ln^{(T)} M \leq 8\sqrt{C_D}/\epsilon_0 $. We run Algorithm 1 $T$ times with parameters $\epsilon_i$ in the $i$-th iteration.

For all values of $i$ this implies that:
\begin{align}
    8\sqrt{C_D}/\epsilon_i = 8\sqrt{C_D}\epsilon_0^{-1} 2^{T-i}
    \leq 2^{T-i} \ln^{(T-1)} M \leq \ln^{(i-1)} M, 
\end{align}
where the first inequality is due to the definition of $T$ (that is, for $T-1$ the inequality holds in the other direction), and the second is due to the fact that for all $x>0$ it holds that $2\ln x < x$.

Using induction we show that for all $i \leq T$ we have that
$$M_i \leq 16C_D(\epsilon^{-1}_i \ln^{(i)}M)^2.$$
For $i=1$ the claim holds by Theorem~\ref{thm: theorem 4.6} (and because $\epsilon_0^{-1} > e$). For the remaining values of $i$, due to Theorem~\ref{thm: theorem 4.6} it holds that $M_i \leq C_D (\epsilon_{i}^{-1} \cdot \log_{1/ \epsilon_i}{M_{i-1}})^2$. Thus, it is enough to show that $\log_{1/ \epsilon_i}{M_{i-1}} \leq 4\ln^{(i)}M$. It holds that
\begin{align*}
   \log_{1/ \epsilon_i}{M_{i-1}} &\leq_{(1)} \ln{M_{i-1}}
    \\ &\leq \ln{16C_D(\epsilon_{i-1}^{-1} \ln^{(i-1)} M)^2} = \ln{16C_D(2\epsilon_{i}^{-1} \ln^{(i-1)} M)^2}
    \\ &\leq_{(2)}
    \ln{(\ln^{(i-1)}  M)^4} \leq 4 \ln^{(i)} M, 
\end{align*}
where the first inequality is due to the fact that $\epsilon_0^{-1} \geq e$ and the third is due to inequality (1). Finally plugging in $i=T$ we get that: $M_T \leq 16C_D\epsilon_0^{-1} \ln^{(T)} M \leq 16C_D/\epsilon_0 \cdot 8\sqrt{C_D}/\epsilon_0 = O(\epsilon^{-2})$.
As the the weighted $\epsilon$-defect for every vertex $v$ is bounded by $w(v)\sum_{i=0}^{T}\epsilon_i \leq \epsilon w(v)$, this completes the proof.
\end{proof}

\section{Max 2-SAT}
\label{sec:max 2-sat}
In this section we consider the problem of Max 2-SAT. We are given a set of weighted clauses $\mathcal{C}$ over the set of node variables, such that each clause contains at most two literals. We wish to find an assignment maximizing the weight of satisfied clauses. As before, we are interested in adapting a sequential algorithm to the distributed setting.
The algorithm we shall adapt is the sequential algorithm presented in \cite{PoloczekSWZ17} which achieves a 3/4-approximation in expectation for the problem of weighted Max-SAT. It is based on the results of \cite{BuchbinderFNS15}, thus it is almost identical to Algorithm~\ref{alg: rand-submodular}. Before presenting the algorithm we need some preliminary definitions.

We allow the node variables to take on values in $\set{0,1,\emptyset}$ where $X_v=\emptyset$ means that the value to $X_v$ has not yet been assigned. We define two utility functions, the first, $f_T$ counts the weight of clauses satisfied given the assignment, and the second $f_F$ counts the weight of clauses that are unsatisfied given the assignment (all literals are false). We note that until now our functions depended on the node variables and the topology of the graph (which we omitted as a parameter), and now they also depend on the clauses. As with the topology, when we pass $\Xbar$ as a parameter we assume we also pass all of the clauses, while when we pass $\Xnv$ as a parameter, we assume we also pass all of the clauses that contain $X_v$ as a literal.

Denote by $w(C)$ the weight of some $C \in \mathcal{C}$. For any $C \in \mathcal{C}$ we define two values, $C_T(\Xbar), C_F(\Xbar)$, where $C_T(\Xbar)=w(C)$ if the clause is satisfied by $\Xbar$ (and 0 otherwise) while $C_F(\Xbar)=w(C)$ if the clause is falsified by the assignment (and 0 otherwise). It holds that $f_T(\Xbar) = \sum_{C \in \mathcal{C}} C_T(\Xbar)$ and $f_F(\Xbar) = \sum_{C \in \mathcal{C}} C_F(\Xbar)$.
Both functions are local utility functions. We prove the following lemma.
\begin{lemma}
\label{lem: max 2-sat local utility}
Both $f_T$ and $f_F$ are local utility functions.
\end{lemma}
\begin{proof}
Let us assign each clause $C$ to one of the variables which appears in it as a literal (perhaps there is only a single one).
Denote for each variable by $C_v$ the set of clauses assigned to it.
Because this is a 2-SAT instance, it holds that if $C \in C_v$ then $C_T(\Xbar) = C_T(L_v[\Xbar])$ and $C_F(\Xbar) = C_F(L_v[\Xbar])$.

It holds that $f_T(\Xbar) = \sum_{v\in V} \sum_{C \in C_v} C_T(L_v[\Xbar])$.
Thus it holds that:
\begin{align*}
&f_T(\Xbar \cup \set{X_v = \alpha'}) - f_T(\Xbar \cup \set{X_v = \alpha})
\\ &= \sum_{u\in N(v)} \sum_{C \in C_u} C_T(L_u[\Xbar]\cup \set{X_v = \alpha'}) - C_T(L_u[\Xbar]\cup \set{X_v = \alpha})
\\ &= \sum_{u\in N(v)} \sum_{C \in C_u} C_T(L_v[\Xbar]\cup \set{X_v = \alpha'}) - C_T(L_v[\Xbar]\cup \set{X_v = \alpha}) \triangleq g_v(L_v[\Xbar,\alpha',\alpha]), 
\end{align*}
where the first equality is because all clauses that do not contain $X_v$ as a literal immediately get removed from the sum, and the final equality is because the clauses are of size at most two, thus every clause that is affected by an assignment to $X_v$ only depends on $X_v$ and $X_u$ s.t $u\in N(v)$, so we may state the difference as a local function. This finishes the proof for $f_T$. The proof for $f_F$ is identical.
\end{proof}

Now we describe the algorithm of \cite{PoloczekSWZ17} (Algorithm~\ref{alg: rand-max-2sat}). In \cite{PoloczekSWZ17}, it is shown that this algorithm achieves a $3/4$ approximation for weighted Max-SAT. In the distributed setting we are interested in the special case of 2-SAT.

Initially all variables are unassigned. The algorithm iterates over the variables in any order. Our aim in each iteration is to maximize the difference between the weight of the satisfied clauses and that of the unsatisfied clauses given the partial assignment. The algorithm calculates two quantities $t_i,f_i$, which correspond to this difference if we set the current variable to 0 or 1, given the current partial assignment. Finally, we flip a coin and set the variable to each value with a probability proportional to its benefit. If the benefit is negative, this probability is 0.

\RestyleAlgo{boxruled}
\LinesNumbered
\begin{algorithm}[htbp]
	
	\caption{\texttt{rand-MaxSAT($G, \mathcal{C}$)}}
	\label{alg: rand-max-2sat}
    $\Xbar = \set{ X_v = \emptyset \mid v \in V}$\\
	Define any order over the variables: $X_{v_1},...,X_{v_n}$
	\For{$i$ from 1 to $n$ }
	{
	    $t_i = (f_T(\Xbar \cup \set{X_{v_i} = 1}) - f_T(\Xbar)) - (f_F(\Xbar \cup \set{X_{v_i} = 1}) - f_F(\Xbar))$\\
        $f_i = (f_T(\Xbar \cup \set{X_{v_i} = 0}) - f_T(\Xbar)) - (f_F(\Xbar \cup \set{X_{v_i} = 0}) - f_F(\Xbar))$\\
        \If{$f_i \leq 0$}
        {
        	$p = 1$\\
        }
        \ElseIf{$t_i \leq 0$}
        {
            $p = 0$\\
        }
        \Else
        {
        	$p = t_i / (t_i + f_i)$\\
        }
        Flip a coin with probability $p$\\
        \If{heads}
        {
            $X_{v_i} = 1$
        }
        \Else
        {
            $X_{v_i} = 0$
        }
    }
    return $\Xbar$\\
\end{algorithm}

Due to Lemma~\ref{lem: max 2-sat local utility} we see that calculating $t_i,f_i$ can be done locally, and thus Algorithm~\ref{alg: rand-max-2sat} is an orderless-local algorithm. Thus we can immidiatly state the following theorem:

\begin{theorem}
\label{thm: max-sat color}
For any graph $G$ and a Weighted Max 2-SAT instance, there exists a randomized distributed expected 3/4-approximation algorithm running in $O(c)$ communication rounds in the CONGEST model, given a legal $c$-coloring.
\end{theorem}

We now show that the same technique used for the Max-DiCut problem can be applied here. We consider the distributed algorithm guaranteed from Theorem~\ref{thm: max-sat color}. We note that in order to compute $t_i,f_i$ for some node $v$ we only need to know the variables in the 1-hop environment that share a clause with $v$. Thus, if there is an edge $(u,v) \in E$ but no clause depends on both $X_v, X_u$ we may drop that edge. Formally, if we define by $H(V,E_H)$ the graph resulting from $G$ by dropping all said edges, the execution of the algorithm for $H$ is identical to that of $G$. And now we may use the same method as before of coloring the vertices at random with colors from $[\ceil{\epsilon^{-1}}]$.

Let us ignore nodes in $H$ which have no edges. These nodes may only contain unit clauses (clauses with a single literal), and the solution achieved by the greedy algorithm for that subset of clauses is optimal. This is because if a node only contains a single literal, the algorithm will satisfy that literal. If it contains two literals which contradict each other, any assignment satisfies one of them, which is optimal.

Let us now assign every clause containing two literals to the edge between the two nodes, and any unit clause to one of the node's edges, denote this set of clauses by $C_e$ for every edge $e\in E_H$. Next we define for every $e\in E_H$ its weight to be $w(e) = \sum_{C \in C_e} w(C)$, thus $\sum_{e\in E_H} w(e) = \sum_{C\in \mathcal{C}} w(c) = w(\mathcal{C})$.


From here the proof follows the lines of Section~\ref{sec: fast cuts}.
Let $H'$ be the graph resulting from $H$ after removing all monochromatic edges and all clauses associated with those edges. This results in a an expected decrease of $\epsilon w(\mathcal{C})$ in total clause weight. We prove a lemma equivalent to Lemma~\ref{lem: drop edges approx} in Section~\ref{sec: fast cuts}.

\begin{lemma}
\label{lem: drop edges approx 2-sat}
Let $H'(V,E_H')$ be a graph resulting from removing any subset of edges from $H$ of weight at most $\epsilon w(\mathcal{C})$. Then for any constant $p$, any $p$-approximation for Max 2-SAT for $H'$ is a $p(1-2\epsilon)$-approximation for $H$.
\end{lemma}
\begin{proof}
Let $OPT, OPT'$ be optimal solutions for $H, H'$, respectively. We can immediately see that $OPT' \geq OPT - \epsilon w(\mathcal{C})$, as any solution for $H$ is also a solution for $H'$ whose value differs by at most $\epsilon w(\mathcal{C})$.

Using the probabilistic method implies that for Weighted Max 2-SAT it holds that $OPT \geq w(\mathcal{C})/2$. This is because a single true literal is enough to satisfy a clause, thus setting the variable values uniformly at random guarantees that every clause has probability at least 1/2 of being satisfied (it equals 1/2 when the clause is a unit clause). Thus, in expectation every clause contributes half its weight to the sum. This guarantees that there exists some solution of value at least $w(\mathcal{C})/2$, and by definition it holds that $OPT \geq w(\mathcal{C})/2$.

  Using all of the above we can say that given a $p$-approximate solution for $OPT'$ it holds that:
$$p\cdot OPT' \geq p(OPT - \epsilon w(\mathcal{C})) \geq  p(OPT - 2\epsilon OPT) = p(1-2\epsilon)OPT $$

\end{proof}

Finally the proof of the expected approximation value follows exactly as in Section~\ref{sec: fast cuts}. We state the main theorem for this section.

\begin{theorem}
\label{thm: defective dicut}
There exists a randomized distributed $(3/4-\epsilon)$-approximation for Weighted Max 2-SAT running in $O(\epsilon^{-1})$ communication rounds in the CONGEST model.
\end{theorem}

\end{document}